\documentclass[12pt,letterpaper]{article}
\usepackage{amsmath,amsthm,amssymb,latexsym,amsfonts,amscd,bbm}
\usepackage[top=1in,left=1in,right=1in,marginpar=2cm]{geometry}

\usepackage{dsfont}
\usepackage{mathrsfs}
\usepackage{setspace}
\usepackage{hyperref}
\usepackage{color}
\usepackage{mathtools}
\usepackage{comment}
\usepackage{tikz}
\usepackage{enumerate}

\newcommand{\ud}{\mathrm{d}}

\newcommand{\e}{\mathrm{e}}

\newcommand{\honeparticlehamiltonian}{h_N}
\newcommand{\honeparticlehamiltonianUnshifed}{\widehat h_N}
\newcommand{\honeparticlehamiltonianEigenvalues}{e_N}
\newcommand{\honeparticlehamiltonianUnshiftedEigenvalues}{\widehat e_N}
\newcommand{\freelaplacianEigenvalues}{e_N}
\newcommand{\freelaplacianEigenvalueOne}{\freelaplacianEigenvalues^{1,\omega}}
\newcommand{\freelaplacianEigenvalueTwo}{\freelaplacianEigenvalues^{2,\omega}}
\newcommand{\freelaplacianEigenvalueThree}{\freelaplacianEigenvalues^{3,\omega}}

\numberwithin{equation}{section}
\newtheorem{theorem}{Theorem}[section]
\newtheorem{lemma}[theorem]{Lemma}
\newtheorem{prop}[theorem]{Proposition}
\newtheorem{cor}[theorem]{Corollary}
\newtheorem{remark}[theorem]{Remark}

\theoremstyle{definition}
\newtheorem{defn}[theorem]{Definition}

\allowdisplaybreaks

\DeclarePairedDelimiterX\braket[2]{\langle}{\rangle}{#1 , #2}

\begin{document}
	
	\thispagestyle{empty}
	
	\vspace*{1cm}
	
	\begin{center}
		
		{\Large \bf On Bose--Einstein condensation in interacting Bose gases in the Kac--Luttinger model} \\

		\vspace*{2cm}
		
		{\large Chiara Boccato \footnote{E-mail address: {\tt chiara.boccato@unimi.it}}}  %
		
		\vspace*{5mm}
		
		Department of Mathematics\\
		Università degli Studi di Milano\\
		20133 Milano \\
		Italy
		
		\vspace*{5mm}
		
		{\large  Joachim~Kerner \footnote{E-mail address: {\tt joachim.kerner@fernuni-hagen.de}}}  %
		
		\vspace*{5mm}
		
		Department of Mathematics and Computer Science\\
		FernUniversit\"{a}t in Hagen\\
		58084 Hagen\\
		Germany

		\vspace*{5mm}

		{\large  Maximilian Pechmann \footnote{E-mail address: {\tt mpechmann@tntech.edu}}}%
		
		\vspace*{5mm}
		
		Department of Mathematics\\
		Tennessee Technological University\\
		Cookeville, TN 38505\\
		USA
		
	\end{center}
	
	\vfill
	
	\begin{abstract} \noindent We study interacting Bose gases of dimensions $2 \le d \in \mathds N$ at zero temperature in a random model known as the Kac--Luttinger model. Choosing the pair-interaction between the bosons to be of a mean-field type, we prove (complete) Bose--Einstein condensation in probability or with probability almost one into the minimizer of a Hartree-type functional. We accomplish this by building upon very recent results by Alain-Sol Sznitman on the spectral gap of the noninteracting Bose gas.
	\end{abstract}
	
	\newpage
	
	\section{Introduction} \label{section introduction}
	
    A gas of bosonic particles at low temperature may exhibit a quantum phenomenon known as Bose--Einstein condensation (BEC), meaning a macroscopic occupation of a one-particle quantum state. Originally predicted by \cite{Bose24,EinsteinBECI,EinsteinBECII} in a noninteracting setting, BEC is expected to occur under suitable conditions in interacting systems as well. A first rigorous proof of BEC in interacting dilute Bose gases has been achieved in a low density scaling regime called the Gross--Pitaevskii limit \cite{LiebSeiringer2002}. In the same regime it has been possible to obtain expressions for the  excitation energies over the condensate and to resolve the corresponding eigenfunctions \cite{BBCS2019}, confirming the predictions of Bogoliubov theory. Similar results have been obtained in mean field limits \cite{Seiringer2011,GrechSeiringer2013} of high density and weak interactions. It is interesting to note that in the above-mentioned regimes the collective behavior of the many-body system is captured by suitable effective one-particle theories, such as Gross-Pitaevskii theory or Hartree theory. These are nonlinear theories, in contrast with the linear underlying microscopic description.	 
    It is a natural to ask whether similar properties are stable in presence of randomly placed impurities, both in the noninteracting and in the interacting setting.
	
	The goal of this paper is to prove BEC in an interacting Bose gas in $\mathds R^d$, $2 \le d \in \mathds N$, placed in a random environment known as the Kac--Luttinger model, originally considered in \cite{kac1973bose,kac1974bose}. In those papers, Kac and Luttinger studied a system of noninteracting bosons in $\mathds R^3$ and with an external potential that is generated by a collection of infinitely many and randomly (according to a Poisson point process) placed hard balls. The key feature of such random systems, which makes them interesting in the context of BEC, is the existence of so-called Lifshitz tails at the bottom of the spectrum~\cite{pastur1992spectra}. A Lifshitz tail refers to an exponentially fast decaying density of states at low energies, and enhances the existence of BEC, at least in a gas of noninteracting bosons. This phenomenon is maybe even more transparent in the one-dimensional analog of the Kac--Luttinger model -- the so-called Luttinger--Sy model~\cite{LuttSyEnergy,GredeskulPastur75}. In any case, similar to what Einstein had observed for the three-dimensional Bose gas without external potential, the smallness of the density of states at the bottom of the spectrum leads to a \textit{finite} critical (particle) density and therefore to some sort of condensation. However, it is much more difficult to determine the actual nature of the condensate or, more precisely, its so-called type. The most classical notion is that of a type-I BEC, which means that only the one-particle ground state is macroscopically occupied and indeed, this is exactly what Kac and Luttinger conjectured for their random model. The proof of this conjecture was achieved only very recently by Alain-Sol Sznitman in \cite{SZNITMAN2023104197} in connection with results obtained in~\cite{KERNER2020287}, see also~\cite{kerner2023mini}. Consequently, due to those findings, the condensate in the \textit{noninteracting} Bose gas in the Kac--Luttinger model is by now well-understood. 
	
	In this paper, our goal is to introduce repulsive two-particle interactions and to prove BEC in the interacting (Kac--Luttinger) model. As mentioned above, due to the presence of Lifshitz tails, BEC is in some sense more stable in random environments, at least for a system of noninteracting bosons; hence one might expect it to be an easy task to allow for repulsive interactions without destroying the condensate. However, this turns out not to be the case, and the reason being is that the eigenfunctions of the underlying one-particle Schrödinger operator are highly localized. In other words, the bosons are spatially more close to each other and hence any strong enough repulsive interaction tends to destroy the condensate immediately as demonstrated in \cite{kerner_pechmann_2023} (for comparable results for the one-dimensional Luttinger--Sy model we refer to \cite{kerner2021effect}). Consequently, the results obtained in~\cite{kerner_pechmann_2023} imply that a one-particle state can be macroscopically occupied only if it is not too localized or if the two-particle interactions are weak enough. In the present paper, we will focus on the second aspect and show that (complete) BEC, in probability or with probability almost one depending on the strength of the interaction, occurs into a localized one-particle state (which turns out to be a minimizer of a certain nonlinear functional) for two-particle interactions that scale with the volume of the one-particle configurations space and tend to zero fast enough in the thermodynamic limit. In other words, we are able to prove (complete) BEC, in probability and with probability almost one in suitable (mean-field) scaling limits for the Kac--Luttinger model in dimension $2 \le d \in \mathds N$.

	The paper is organized as follows: In Section~\ref{section model} we introduce the random Kac--Luttinger model and the $N$-particle Hamiltonian. In Section~\ref{section hartree functional} we introduce a Hartree functional and derive auxiliary results. This will then allow us to prove condensation in Section~\ref{section proof bec}.  
	
	\section{The model} \label{section model}
	\subsection{The underlying random model}
	We consider a $d$-dimensional system, $2 \le d \in \mathds N$, of $N$ bosons, $N \in \mathds N$, in the box
	\begin{equation}
	\Lambda_N \coloneqq \left(-L_N/2,+L_N /2 \right)^d \subset \mathds{R}^d \quad \text{ with } \quad L_N = \rho^{-1/d} N^{1/d}
	\end{equation}
	for all $N \in \mathds N$. Here, $\rho > 0$ denotes the particle density. This means that the limit $N \rightarrow \infty$ refers to the standard thermodynamic limit. 
	
	The random model to be discussed employs an external potential $V$ on $\mathds R^d$ that is informally defined by (given a probability space $(\Omega, \mathcal F, \mathds P)$)
	\begin{equation}
	V : \Omega \times \mathds R^d \to \mathds R \cup \infty, \quad (\omega, x) \mapsto V^{\omega}(x) \coloneqq \sum\limits_{m} \infty \cdot B_{r}(x - x_m^{\omega}) \ .
	\end{equation}
	Here, the set $\{x_m^{\omega}\}_{m}$ is generated by a Poisson point process on $\mathds R^d$ with a constant intensity $\nu > 0$, and $B_{r}(x)$ is a ball of fixed radius $r>0$ with center $x \in \mathds R^d$.
	We denote the random domain that is generated by the random external potential $V$ by
	\begin{equation} 
	\Lambda^{\omega}_{N} \coloneqq \Lambda_N \backslash \bigcup\limits_m B_{r}(x_m^{\omega}), \quad \omega \in \Omega, \ N \in \mathds N 
	\end{equation}
	and call $\Lambda^{\omega}_{N}$ the vacancy set. We remark that the volume of $\Lambda_N^{\omega}$ tends to be a constant fraction of $\Lambda_N$ in the limit $N \to \infty$. More precisely, we have $\lim_{N \to \infty} \mathds P(\Omega_N^{(1),\eta}) = 1$ for any $\eta > 0$ where 
	\begin{equation}
	\Omega_N^{(1),\eta} \coloneqq \left\{ \omega \in \Omega : \left| |\Lambda_N^{\omega}|/|\Lambda_N| - \e^{-\nu \omega_d r^d} \right| < \eta \right\}
	\end{equation}
	and $\omega_d$ is the volume of the $d$-dimensional unit ball in $d$ \cite[p. 147]{sznitman1998brownian}.
	Also, the vacancy set $\Lambda_N^{\omega}$ may be divided into non-empty connected components (regions). However, $\Lambda_N^{\omega}$ has $\mathbb{P}$-almost surely only finitely many components for each $N \in \mathbb N$ \cite[Proposition 4.1]{meester1996continuum}, and from now on we consider only $\omega \in \Omega$ with that property. 
	
	For each $N \in \mathds N$, we denote the number of these components by $K_N^{\omega} \in \mathds N_0$ (with the understanding that $K_N^{\omega} \coloneqq 0$ whenever $\Lambda_N^{\omega} = \emptyset$.) Note that if $0 < \eta < \e^{-\nu \omega_d r^d}$, then for any $N \in \mathds N$ and any $\omega \in \Omega_N^{(1),\eta}$, we have $\Lambda_N^{\omega} \neq \emptyset$ and thus $K_N^{\omega} \ge 1$. We set $\mathcal K_N^{\omega} := \left\{1, \ldots, K_N^{\omega} \right\}$ if $K_N^{\omega} \ge 1$ (and $\mathcal K_N^{\omega} := \emptyset$ if $K_N^{\omega} = 0$) and label the components by $k \in \mathcal K_N^{\omega}$. Hence, we can denote each component of $\Lambda_N^{\omega}$ by $\Lambda_N^{k,\omega}$, $k \in \mathcal K_N^{\omega}$. Note that $\{\Lambda_N^{k,\omega}\}_{k \in \mathcal K_N^{\omega}}$ is a partition of $\Lambda_N^{\omega}$.
	
	\subsection{The many-particle Hamiltonian}
	
	For introducing the $N$-particle Hamiltonian, let
	\begin{equation}
	v_N : \mathds R^d \to \mathds R, \quad x \mapsto v_N(x) , \quad N \in \mathds N\ ,
	\end{equation}
	be a potential describing the interaction between two bosons. We shall assume that $v_N \in (L^1 \cap L^{\infty})(\mathds R^d)$ is a nonnegative, even, positive-definite (meaning that the Fourier transform $\widehat v_N$ of $v_N$ is nonnegative) function such that $\widehat v_N \in L^1(\mathds R^d)$ for all $N \in \mathds N$.
	
	Therefore, for any $0 < \eta < \e^{-\nu \omega_d r^d}$, all $N \in \mathds N$, and all $\omega \in \Omega_N^{(1),\eta}$, our system is described by the random, self-adjoint $N$-particle Hamiltonian
	\begin{equation} \label{N particle Hamiltonian}
	H_{N}^{\omega} \coloneqq - \sum_{j=1}^{N} \Delta_{j} + \sum_{1\leq i < j \leq N} v_N(x_i-x_j)
	\end{equation}
	defined on $L_{\text{s}}^2((\Lambda^{\omega}_{N})^N)$, employing Dirichlet-boundary conditions along the boundary of $\Lambda^{\omega}_N$. 
	We remark that the index $\text{s}$ refers to the totally symmetric subspace of $L^2((\Lambda^{\omega}_{N})^N)$. Moreover, the form domain of $H_N^{\omega}$ is given by $\mathcal D[H_N^{\omega}] = H^1_0(\Lambda^{\omega}_{N})$. Consequently, the ground state energy $E_{\text{QM},N}^{1,\omega}$ of $H_N^{\omega}$, that is, the lowest eigenvalue of $H_N^{\omega}$ is determined via
	\begin{equation}
	E_{\text{QM},N}^{1,\omega} \coloneqq \inf \left\{ \braket{\psi}{ H_N^{\omega} \psi} : \psi \in \mathcal D[H_N^{\omega}] \text{ and } \|\psi \|_{L^2(\mathds R^d)} = 1 \right\} \ ,
	\end{equation}
	where $ \braket{\psi}{ H_N^{\omega} \psi}$ is understood in the form sense. Lastly, for any $\eta > 0$, all $N \in \mathds N$, and all $\omega \notin \Omega_N^{(1),\eta}$, we set $H_N^{\omega} \coloneqq 0$. 
	
	\begin{remark}
		In order to keep the notation simpler in the following, we shall abbreviate $L^p$-norms by writing, for example, $\|\cdot\|_2$ instead of $\|\cdot\|_{L^2(\Lambda_N^{\omega})}$ or, similarly,  $\|\cdot\|_1$ instead of $\|\cdot\|_{L^1(\mathds R^d)}$. Hence, we neglect the actual domains of integration whenever its clear from the context.
	\end{remark}
	
	\begin{remark}
        Throughout this work, we will use the following notation regarding the limiting behaviour of sequences. Let $(a_N)_{N \in \mathds N}$, $(b_N)_{N \in \mathds N}$ be two sequences. Then $a_N \ll b_N$ if and only if $\lim_{N \to \infty} a_N / b_N = 0$. Furthermore, $a_N \sim b_N$ if and only if there are constants $c, C \in \mathds R$ such that $c a_N \le b_N \le C a_N$ for all but finitely many $N \in \mathds N$. Lastly, we write $a_N \lesssim b_N$ if and only if $a_N \ll b_N$ or $a_N \sim b_N$.
	\end{remark}
	
	We assume the particles to be weakly interacting in a suitable sense. More explicitly, for our main result, Theorem~\ref{theorem proof of BEC II}, we assume the interaction potential $v_N$ to scale with $N$ such that $\|v_N\|_1 \lesssim N^{-1}(\ln N)^{-2/d}$ and $v_N(0) \ll (\ln N)^{-(1+2/d)}$. For example, $v_N$ can be such that
	\begin{equation} \label{eq:vN}
	v_N(x)=\frac{\kappa V(x)}{N (\ln N)^{2/d}}
	\end{equation}
	where $V \in (L^1 \cap L^{\infty})(\mathds R^d)$ has support on a set independent on $N$, and the coupling constant $\kappa>0$ is sufficiently small.
	
	\begin{remark}
	A scaling of the interaction as in \eqref{eq:vN} is of particular relevance. To see this, note that with a probability that converges to one, the lowest eigenfunction of the Dirichlet Laplacian on $\Lambda_N^{\omega}$ is supported only on a single component of $\Lambda_N^\omega$, as we will prove in Lemma~\ref{Lemma ground state of Dirichlet Laplace has support only on one single component} in combination with Proposition~\ref{proposition probability omegaN2 converges to one}. In the nonpercolation regime (meaning that the intensity of the Poissonian point process is sufficiently large), this component has a volume bounded from above by $(const.) \ln N$, see for example \cite{kerner_pechmann_2023}. As long as the interaction is not too strong, we therefore expect the particles to be effectively localized in a volume at most of order $\ln N$. This leads to a particle density of at least $\sim N/\ln N$ in that component. Moreover, as shown in \cite{SZNITMAN2023104197}, with probability arbitrarily close to one as $N\to\infty$, the spectral gap of the Dirichlet Laplacian always stays bigger than $\sigma (\ln N)^{-(1+2/d)}$, where $\sigma$ is a small positive number (the smaller $\sigma$, the closer the probability is to one). Therefore, an interaction strength such as \eqref{eq:vN} leads to a potential energy per particle of order $\kappa (\ln N)^{-(1+2/d)}$, which is comparable in size to the spectral gap of the Dirichlet Laplacian.
	\end{remark}
	
	\section{Hartree-type functionals} \label{section hartree functional}
	In order to establish existence of Bose--Einstein condensation in Section~\ref{section proof bec}, we introduce the one-particle Hartree-type functionals
	\begin{equation}\label{HartreeFunctional componentwise}
	\mathcal{E}_N^{k,\omega}[\psi] \coloneqq \int\limits_{\Lambda_N^{k,\omega}} |\nabla \psi(x)|^2 \ \ud x + \frac{N-1}{2} \int\limits_{\Lambda_N^{k,\omega}}\int\limits_{\Lambda_N^{k,\omega}} v_N(x-y)|\psi(x)|^2|\psi(y)|^2\ \ud x \ud y
	\end{equation}
	with domain $\mathcal D(\mathcal E_N^{k,\omega}) \coloneqq H_0^1(\Lambda_N^{k,\omega})$
	for any $0 < \eta < \e^{-\nu \omega_d r^d}$, all $N \in \mathds N$, all $\omega \in \Omega_N^{(1),\eta}$, and all $k \in \mathcal K_N^{\omega}$. Note that the domain of $\mathcal{E}_N^{k,\omega}[\psi]$ only includes functions that are supported only on a single component of $\Lambda_N^{\omega}$.
	\begin{defn}
		Let an arbitrary $0 < \eta < \e^{-\nu \omega_d r^d}$, $N \in \mathds N$, $\omega \in \Omega_N^{(1),\eta}$, and $u \in H^1_0(\Lambda_N^{\omega})$ be given. We introduce the linear, self-adjoint, one-particle Hamiltonian $\honeparticlehamiltonian^{u,\omega}$ on $L^2(\Lambda_N^{\omega}) = \bigoplus_{k \in \mathcal K_N^{\omega}} L^2(\Lambda_N^{k,\omega})$ by
		\begin{align} \label{definition klein widetilde h N u omega}
		\honeparticlehamiltonian^{u,\omega} \coloneqq -\Delta + (N-1) (|u|^2 \ast v_N) - \frac{N-1}{2}\int\limits_{\Lambda_N^{\omega}} \int\limits_{\Lambda_N^{\omega}}v_N(x-y)|u(x)|^2|u(y)|^2\ \ud x \ud y
		\end{align}
		with form domain $\mathcal D[\honeparticlehamiltonian^{u,\omega}] \coloneqq H_0^1(\Lambda_N^{\omega}) = \bigoplus_{k \in \mathcal K_N^{\omega}} H_0^1(\Lambda_N^{k,\omega})$. Here, $*$ denotes the convolution of two functions. We write $\honeparticlehamiltonianEigenvalues^{1,u,\omega}$ and $\honeparticlehamiltonianEigenvalues^{2,u,\omega}$ for the lowest and second-lowest eigenvalue of $\honeparticlehamiltonian^{u,\omega}$, respectively, counting with multiplicity.
		
		In addition,  for any $k \in \mathcal K_N^{\omega}$ we define the linear, self-adjoint, one-particle Hamiltonian $\honeparticlehamiltonian^{u, k, \omega}$ in $L^2(\Lambda_N^{k,\omega})$ by
		\begin{equation} \label{definition klein widetilde h N u k omega}
		\honeparticlehamiltonian^{u, k, \omega} \coloneqq -\Delta + (N-1) (|u|^2 \ast v_N)- \frac{N-1}{2}\int\limits_{\Lambda_N^{\omega}} \int\limits_{\Lambda_N^{\omega}}v_N(x-y)|u(x)|^2|u(y)|^2\ \ud x \ud y
		\end{equation}
		with form domain $\mathcal D[\honeparticlehamiltonian^{u,k,\omega}] \coloneqq H_0^1(\Lambda_N^{k,\omega})$. Similarly, we denote the lowest and second-lowest eigenvalue of $\honeparticlehamiltonian^{u,k,\omega}$ by $\honeparticlehamiltonianEigenvalues^{1,u,k,\omega}$ and $\honeparticlehamiltonianEigenvalues^{2,u,k,\omega}$, respectively.
	\end{defn}
    \begin{prop} \label{Proposition ground state energy hNk equals componentwise hartree functional}
		For any $0 < \eta < \e^{-\nu \omega_d r^d}$, all $N \in \mathds N$, all $\omega \in \Omega_N^{(1),\eta}$, and all $k \in \mathcal K_N^{\omega}$, the functional $\mathcal{E}_N^{k,\omega}$ has (up to a phase) a unique, real-valued, positive minimizer $u_N^{k,\omega} \in H_0^1(\Lambda_N^{k,\omega})$ with $\|u_N^{k,\omega}\|_2 = 1$
		and corresponding energy $\varepsilon_N^{1,k,\omega}$,
		\begin{equation}  \label{definition widehat eN1omega}
		\mathcal{E}_N^{k,\omega}[u_N^{k,\omega}] = \varepsilon_N^{1,k,\omega} \coloneqq \min \left\{ \mathcal E_N^{k,\omega}[\psi] : \psi \in H_0^1(\Lambda_N^{k,\omega}), \|\psi\|_2 = 1 \right\}  \ge 0\ .
		\end{equation}
		Moreover, $u_N^{k,\omega}$ and $\varepsilon_N^{1,k,\omega}$ are also the ground state and the ground-state energy, respectively, of $\honeparticlehamiltonian^{u_N^{k,\omega}, k, \omega}$,
		\begin{equation}
		\honeparticlehamiltonian^{u_N^{k,\omega}, k, \omega} u_N^{k,\omega} = \honeparticlehamiltonianEigenvalues^{1,u_N^{k,\omega}, k,\omega} u_N^{k,\omega}
		\end{equation}
		where
		\begin{equation}
		\honeparticlehamiltonianEigenvalues^{1,u_N^{k,\omega}, k,\omega} \coloneqq \min \left\{ \braket{\psi}{ \honeparticlehamiltonian^{u_N^{k,\omega}, k, \omega} \psi} : \psi \in H_0^1(\Lambda_N^{k,\omega}) \text{ and } \|\psi\|_{2} = 1 \right\}  = \varepsilon_N^{1,k,\omega} \ .
		\end{equation}
	\end{prop}
	\begin{proof} %\textcolor{red}{Proof needs to be rewriten for $\widetilde h$ instead of $\widehat h$!!!}
		The proof of existence is fairly standard but we include it for completeness: Let an arbitrary $0 < \eta < \e^{-\nu \omega_d r^d}$, $N \in \mathds N$, $\omega \in \Omega_N^{(1),\eta}$, and $k \in \mathcal K_N^{\omega}$ be given. To prove existence of a normalized minimizer $\widetilde u_N^{k,\omega} \in H_0^1(\Lambda_N^{k,\omega})$ of the functional \eqref{HartreeFunctional componentwise}, one first picks a minimizing sequence $(v_{N,n}^{k,\omega})_{n \in \mathds N}$ of normalized functions $v_{N,n}^{k,\omega} \in H_0^1(\Lambda_N^{k,\omega})$, $n \in \mathds N$. This sequence has, due to boundedness of $\Lambda_N^{k,\omega}$ and the compact embedding $H^1_0(\Lambda_N^{k,\omega}) \hookrightarrow
		L^2(\Lambda_N^{k,\omega})$, a subsequence that converges weakly in $H^1_0(\Lambda_N^{k,\omega})$ and strongly in $L^2(\Lambda_N^{k,\omega})$ to a function $\widetilde u_N^{k,\omega}$. Hence, $\|\widetilde u_N^{k,\omega}\|_2=1$.
		
		Now, for the kinetic part of \eqref{HartreeFunctional componentwise} (meaning the first integral in \eqref{HartreeFunctional componentwise}), one then employs lower semi-continuity of the norm while for the potential term of \eqref{HartreeFunctional componentwise} (the second integral in $\eqref{HartreeFunctional componentwise}$), one utilizes Fatou's lemma to conclude
		\begin{equation}
		\mathcal{E}_N^{k,\omega}[\widetilde u_N^{k,\omega}] \leq \liminf_{j \rightarrow \infty}\mathcal{E}_N^{k,\omega}[v_{N,n_j}^{k,\omega}] 
		\end{equation}
		along a subsequence $(v_{N,n_j}^{k,\omega})_{j \in \mathds N}$. This proves that $\widetilde u_N^{k,\omega}$ is a normalized minimizer.
		
		In a next step, the diamagnetic inequality implies $\mathcal{E}_N^{k,\omega}[\widetilde u_N^{k,\omega}] \geq \mathcal{E}_N^{k,\omega}[|\widetilde u_N^{k,\omega}|]$. Therefore, $u_N^{k,\omega} \coloneqq |\widetilde u_N^{k,\omega}|$ is a real-valued, non-negative and normalized minimizer. 
		
		Moreover, because $u_N^{k,\omega}$ minimizes the functional \eqref{HartreeFunctional componentwise}, it fulfils the Euler--Lagrange equation 
		\begin{equation}\begin{split}
        - \Delta u_N^{k,\omega} + &(N-1) (|u_N^{k,\omega}|^2 \ast v_N) u_N^{k,\omega}\\
		&= \Big(\varepsilon_N^{1,k,\omega} + \frac{N-1}{2} \int\limits_{\Lambda_N^{k,\omega}}\int\limits_{\Lambda_N^{k,\omega}} v_N(x-y)|u_N^{k,\omega}(x)|^2 |u_N^{k,\omega}(y)|^2\ \ud x \ud y \Big) u_N^{k,\omega} \ .
        \end{split}	    
		\end{equation}
		This means that $u_N^{k,\omega}$ is also an eigenfunction of $\honeparticlehamiltonian^{u_N^{k,\omega}, k, \omega}$ corresponding to the eigenvalue $\varepsilon_N^{1,k,\omega}$. Since $u_N^{k,\omega}$ is non-negative, it has to be the ground state of $\honeparticlehamiltonian^{u_N^{k,\omega}, k, \omega}$ and hence is positive and unique~\cite{lieb2001analysis}. Also, since $\widetilde u_N^{k,\omega}$ fulfils the same Euler--Lagrange equation, we conclude that $u_N^{k,\omega}=\widetilde u_N^{k,\omega}$ up to a phase.

        Finally, let us remark on uniqueness. Assuming existence of two different (positive) minimizers $\varphi_1,\varphi_2 \in H^1_0(\Lambda_N^{k,\omega})$, one defines, for $0< t < 1$, $\varphi:=\sqrt{t\varphi^2_1+(1-t)\varphi^2_2} \in H^1_0(\Lambda_N^{k,\omega})$ with the aim to show $\mathcal{E}_N^{k,\omega}[\varphi] < t\mathcal{E}_N^{k,\omega}[\varphi_1] +(1-t)\mathcal{E}_N^{k,\omega}[\varphi]_2=\varepsilon_N^{1,k,\omega}$, leading to a contradiction. To show such an inequality, we can employ the transformation employed in the proof of [Lemma~3.3,\cite{lewin2015}] to conclude a corresponding estimate for the non-linear term but with an $\leq$ sign. In addition, we can employ [Theorem~7.8,\cite{lieb2001analysis}] to conclude that the linear term in~\eqref{HartreeFunctional componentwise} fulfils the desired inequality but with an $<$ sign. From this we conclude uniqueness, taking into account that every minimizer is (up to a phase) positive as concluded above.
	\end{proof}
	
	\begin{defn}
		For any $0 < \eta < \e^{-\nu \omega_d r^d}$, all $N \in \mathds N$, and all $\omega \in \Omega_N^{(1),\eta}$ we denote the eigenvalues of the Dirichlet Laplacian $-\Delta$ in $L^2(\Lambda_N^{\omega})$ with form domain $\mathcal D [- \Delta] = H_0^1(\Lambda_N^{\omega}) = \bigoplus_{k \in \mathcal K_N^{\omega}} H_0^1(\Lambda_N^{k,\omega})$, arranged in increasing order and repeated according to their multiplicities, by $0 < \freelaplacianEigenvalueOne \le \freelaplacianEigenvalueTwo \le \freelaplacianEigenvalueThree \le \ldots$.
		We denote the normalized eigenfunctions corresponding to the two lowest eigenvalues $\freelaplacianEigenvalueOne$ and $\freelaplacianEigenvalueTwo$ by $\phi_N^{1,\omega}$ and $\phi_N^{2,\omega}$, respectively:
		\begin{align}
		- \Delta \phi_N^{1,\omega} & = \freelaplacianEigenvalueOne \phi_N^{1,\omega}
		\shortintertext{and}
		- \Delta \phi_N^{2,\omega} & = \freelaplacianEigenvalueTwo \phi_N^{2,\omega} \ .
		\end{align}
		Lastly, we define the restriction of $-\Delta$ to a single component $\Lambda_N^{k,\omega}$ of $\Lambda_N^{\omega}$, $k \in \mathcal K_N^{\omega}$, by $-\Delta|_{\Lambda_N^{k,\omega}}$, that is, $-\Delta|_{\Lambda_N^{k,\omega}}$ is the Dirichlet Laplacian $-\Delta$ in $L^2(\Lambda_N^{k,\omega})$ with form domain $\mathcal D[-\Delta|_{\Lambda_N^{k,\omega}}] = H_0^1(\Lambda_N^{k,\omega})$.
	\end{defn}
	
	\begin{defn} \label{definition Omega 2 event}
	 For any $0 < \eta < \e^{-\nu \omega_d r^d}$ and $N \in \mathds N$, we define the event
		\begin{equation}
		\Omega_N^{(2),\eta} \coloneqq \left\{ \omega \in \Omega_N^{(1),\eta} : \freelaplacianEigenvalueTwo - \freelaplacianEigenvalueOne > C_1^2 N \|v_N\|_1 (\freelaplacianEigenvalueOne)^{d/2} \right\}
		\end{equation}
		where $C_1 := 2 (4\pi)^{-d/4} \e$.
	\end{defn}
	
	The following Proposition~\ref{proposition probability omegaN2 converges to one} specifies the gap between the two lowest eigenvalues of the Dirichlet Laplacian $-\Delta$ on $\Lambda_N^{\omega}$. In particular, since $e_N^{1,\omega} > 0$, it shows that the ground state of the Dirichlet Laplacian $-\Delta$ on $\Lambda_N^{\omega}$ is unique, with a certain probability. This fact will be important in the proof of Theorem~\ref{theorem proof of BEC I} and~\ref{theorem proof of BEC II}.
	
	\begin{prop} \label{proposition probability omegaN2 converges to one}
		Let an $0 < \eta < \e^{-\nu \omega_d r^d}$ be given.
		\begin{enumerate}[(i)]
		\item For any $\varepsilon > 0$ there exists a constant $\kappa > 0$ such that if $\|v_N\|_1 \le \kappa N^{-1} (\ln N)^{-2/d}$ for all but finitely many $N \in \mathds N$, then 
		\begin{equation}
		\liminf_{N \to \infty} \mathds P (\Omega_N^{(2),\eta}) \ge 1 - \varepsilon \ .
		\end{equation}
		\item If $\|v_N\|_1 \ll N^{-1} (\ln N)^{-2/d}$, we have
		\begin{equation}
            \lim_{N \to \infty} \mathds P (\Omega_N^{(2),\eta}) =1 \ .
		\end{equation}
		\end{enumerate}
	\end{prop}
	\begin{proof}
		Firstly, note that
		\begin{equation} \label{Equation Sznitman Theorem 6.1}
		\lim\limits_{\sigma \to 0} \liminf\limits_{N \to \infty} \mathds P \left( \freelaplacianEigenvalueTwo - \freelaplacianEigenvalueOne \ge \sigma (\ln N)^{-(1+2/d)} \right) = 1 \ ,
		\end{equation}
		see \cite[Theorem~6.1]{SZNITMAN2023104197}.
		Also, there is a nonrandom constant $c>0$ such that almost surely and for all but finitely many $N \in \mathds N$, we have $\freelaplacianEigenvalueOne \le c (\ln N)^{-2/d}$ \cite[Chapter 4, Theorem 4.6]{sznitman1998brownian}. Therefore, there exists a constant $C_2 > 0$ such that $\lim_{N \to \infty} \mathds P ( \Omega_N^{(3)} ) = 1$ where 
		\begin{equation}
		 \Omega_N^{(3),\eta} \coloneqq \left\{ \omega \in \Omega_N^{(1),\eta} :  e_N^{1,\omega} \le C_2 (\ln N)^{-2/d} \right\} \ . 
		\end{equation}

		We firstly discuss case $(i)$: Let an $\varepsilon > 0$ be arbitrarily given. Then due to~\eqref{Equation Sznitman Theorem 6.1}, there exists a $\sigma > 0$ such that
		\begin{equation}
		\liminf\limits_{N \to \infty} \mathds P \left( \freelaplacianEigenvalueTwo - \freelaplacianEigenvalueOne \ge \sigma (\ln N)^{-(1+2/d)} \right) \ge 1 - \varepsilon \ ,
		\end{equation}
		Therefore, if $\|v_N\|_1 \le \kappa N^{-1} (\ln N)^{-2/d}$ for all but finitely many $N \in \mathds N$ and $\kappa \le \sigma C_1^{-2} C_2^{-d/2}$, we have
		\begin{equation}
		\liminf_{N \to \infty} \mathds P (\Omega_N^{(2),\eta}) \ge 1 - \varepsilon \ .
		\end{equation}
		
		As for case $(ii)$, we conclude with \eqref{Equation Sznitman Theorem 6.1} that for any sequence $(\sigma_N)_{N \in \mathds N}$ that converges to zero, we have
		\begin{equation}
		\lim\limits_{N \to \infty} \mathds P \left( \freelaplacianEigenvalueTwo - \freelaplacianEigenvalueOne \ge \sigma_N (\ln N)^{-(1+2/d)} \right) = 1 \ .
		\end{equation}
		We set $\sigma_N \coloneqq C_1^2 C_2^{d/2} N (\ln N)^{d/2} \|v_N\|_1$ for all $N \in \mathds N$. Then $(\sigma_N)_{N \in \mathds N}$ converges to zero and
		\begin{equation}
            \lim_{N \to \infty} \mathds P (\Omega_N^{(2),\eta}) =1 \ .
		\end{equation}
	\end{proof}
   
	In the following lemma we will show that the ground state of the Dirichlet Laplacian is, under suitable assumptions, supported on only one component. This component will then play a crucial role in the subsequent discussion.
	
	\begin{lemma} \label{Lemma ground state of Dirichlet Laplace has support only on one single component}
		Suppose $0 < \eta < \e^{-\nu \omega_d r^d}$, $N \in \mathds N$, and $\omega \in \Omega_N^{(2),\eta}$. Then $\phi_N^{1,\omega}$ has support only on one single component of $\Lambda_N^{\omega}$ \ .
	\end{lemma}
	\begin{proof}
		Let $0 < \eta < e^{-\nu \omega_d r^d}$, $N \in \mathds N$, and an arbitrary $\omega \in \Omega_N^{(1),\eta}$ be given. Suppose that $\phi_N^{1,\omega}$ is supported on more than one component of $\Lambda_N^{\omega}$. We then denote by $\hat k_1^{\omega}$ and $\hat k_2^{\omega}$ two components on which $\phi_N^{1,\omega}$ is supported. Define $\psi_N^{1,\omega} \coloneqq \|\phi_N^{1,\omega} \mathds 1_{\Lambda_N^{\hat k_1^{\omega},\omega}}\|_2^{-1} \phi_N^{1,\omega} \mathds 1_{\Lambda_N^{\hat k_1^{\omega},\omega}}$ and $\psi_N^{2,\omega} \coloneqq \|\phi_N^{1,\omega} \mathds 1_{\Lambda_N^{\hat k_2^{\omega},\omega}}\|_2^{-1} \phi_N^{1,\omega} \mathds 1_{\Lambda_N^{\hat k_2^{\omega},\omega}}$. Now, $\psi_N^{1,\omega}$ and $\psi_N^{2,\omega}$ are both normalized eigenfunctions of $-\Delta$ in $L^2(\Lambda_N)$ that have corresponding eigenvalue $\freelaplacianEigenvalueOne$. Consequently, we would have $\freelaplacianEigenvalueOne = \freelaplacianEigenvalueTwo$, and therefore $\omega \notin \Omega_N^{(2),\eta}$.
	\end{proof}

	\begin{defn} \label{definition widetilde j N omega}
		For any $0 < \eta < \e^{-\nu \omega_d r^d}$, $N \in \mathds N$, and $\omega \in \Omega_N^{(2),\eta}$ we define $\widetilde k_N^{\omega} \in \mathcal K_N^{\omega}$ to be the component $\Lambda_N^{k,\omega}$ on which the normalized eigenfunction $\phi_N^{1,\omega}$ corresponding to the lowest eigenvalue $\freelaplacianEigenvalueOne$ of $-\Delta$ in $L^2(\Lambda_N)$ has its support.
	\end{defn}
	
	\begin{remark}
	 To make our notation easier to read, we define
	 \begin{align}
		 u_N^{\widetilde k,\omega} & \coloneqq u_N^{\widetilde k_N^{\omega},\omega} \ .
    \end{align}
	 We furthermore use $\widetilde k$ instead of $\widetilde k_N^{\omega}$ and $\widetilde u$ instead of $u_N^{\widetilde k_N^{\omega}}$ in the superscriptum whenever it does not lead to confusion. For example, we write $\honeparticlehamiltonianEigenvalues^{1,\widetilde u, \widetilde k, \omega}$ instead of $\honeparticlehamiltonianEigenvalues^{1,u_N^{\widetilde k_N^{\omega},\omega}, \widetilde k_N^{\omega}, \omega}$.
	\end{remark}

	We need the next lemma in the proof of our main result, Theorem~\ref{theorem proof of BEC I}, more precisely in the last step of equation \eqref{eq:ub}.
	
		\begin{lemma} \label{lemma groundstate of h equals lowest componentwise h ground state}
		If $0 < \eta < \e^{-\nu \omega_d r^d}$, $N \in \mathds N$, and $\omega \in \Omega_N^{(2),\eta}$, then
		\begin{equation}
		\honeparticlehamiltonianEigenvalues^{1,\widetilde u, \widetilde k, \omega} = \min \Big\{ \braket{\psi}{\honeparticlehamiltonian^{\widetilde u,\omega} \psi} : \psi \in H_0^1(\Lambda_N^{\omega}) \text{ and } \|\psi\|_2 = 1 \Big\} \ .
		\end{equation}
	\end{lemma}
	\begin{proof}
		Let an $0 < \eta < \e^{-\nu \omega_d r^d}$ and $N \in \mathds N$ be given. Choose an arbitrary $\omega \in \Omega_N^{(2),\eta}$. We show that for any function $\psi \in H_0^1(\Lambda_N^{\omega})$ with $\|\psi\|_2 = 1$ we have
		\begin{equation} \label{Equation Lemma 3.10 asxclke}
		\braket{\psi}{\honeparticlehamiltonian^{\widetilde u, \omega} \psi} \ge \honeparticlehamiltonianEigenvalues^{1,\widetilde u, \widetilde k, \omega} \ .
		\end{equation}
		
		To do this, we show the corresponding version of \eqref{Equation Lemma 3.10 asxclke} for the unshifted analogs of $\honeparticlehamiltonian^{u,\omega}$ defined in \eqref{definition klein widetilde h N u omega} and to $\honeparticlehamiltonian^{u,k,\omega}$ defined in \eqref{definition klein widetilde h N u k omega}. Namely, we consider 
		\begin{align} \label{definition klein widetilde h N u omega unshifted}
		\honeparticlehamiltonianUnshifed^{u,\omega} \coloneqq -\Delta + (N-1) (|u|^2 \ast v_N)
		\end{align}
		and 
		\begin{align} \label{definition klein widetilde h N u k omega unshifted}
		\honeparticlehamiltonianUnshifed^{u,k,\omega} \coloneqq -\Delta + (N-1) (|u|^2 \ast v_N)
		\end{align}
		with the same domains as the associated unshifted operators. 	We denote the lowest eigenvalue of $\honeparticlehamiltonianUnshifed^{u,k,\omega}$ by $\honeparticlehamiltonianUnshiftedEigenvalues^{1,u,k,\omega}$.

		Now, for any function $\psi \in H_0^1(\Lambda_N^{\omega})$ with $\|\psi\|_2 = 1$ we write
		\begin{equation}
		\psi = (1 - \varepsilon)^{1/2} \psi_1 + \varepsilon^{1/2} \psi_2
		\end{equation}
		for $0 \le \varepsilon \le 1$, where
		$\psi_1 := \|\psi \mathds 1_{\Lambda_N^{\widetilde k_N^{\omega}}} \|_2^{-1} \psi \mathds 1_{\Lambda_N^{\widetilde k_N^{\omega}}}$ and $\psi_2 := \|\psi \mathds 1_{\Lambda_N^{\omega} \backslash \Lambda_N^{\widetilde k_N^{\omega}}}\|_2^{-1} \psi \mathds 1_{\Lambda_N^{\omega} \backslash \Lambda_N^{\widetilde k_N^{\omega}}}$ and whenever $\psi\mathds 1_{\Lambda_N^{\widetilde k_N^{\omega}}} \neq 0$. We then have
		\begin{align}
		\begin{split}
		\braket{\psi}{\honeparticlehamiltonianUnshifed^{\widetilde u, \omega} \psi} & = (1-\varepsilon)\braket{\psi_1}{\honeparticlehamiltonianUnshifed^{\widetilde u, \omega} \psi_1} + \varepsilon \braket{\psi_2}{\honeparticlehamiltonianUnshifed^{\widetilde u, \omega} \psi_2} \\
		& \ge (1-\varepsilon) \honeparticlehamiltonianUnshiftedEigenvalues^{1, \widetilde u, \widetilde k, \omega} + \varepsilon \freelaplacianEigenvalueTwo \ .
		\end{split}
		\end{align}
	
  Whenever $\psi\mathds 1_{\Lambda_N^{\widetilde k_N^{\omega}}}=0$, one directly obtains 
  \begin{equation}
      \braket{\psi}{\honeparticlehamiltonianUnshifed^{\widetilde u, \omega} \psi} =  \braket{\psi_2}{\honeparticlehamiltonianUnshifed^{\widetilde u, \omega} \psi_2} \ge \freelaplacianEigenvalueTwo \ .
  \end{equation}
		
		We know claim that $\freelaplacianEigenvalueTwo \ge \honeparticlehamiltonianUnshiftedEigenvalues^{1, \widetilde u, \widetilde k, \omega}$: Indeed, let $\phi_N^{1,\omega}$ be the normalized ground state of $-\Delta$ in $L^2(\Lambda_N^{\omega})$, that is, the normalized eigenfunction corresponding to the eigenvalue $\freelaplacianEigenvalueOne$. Note that $\phi_N^{1,\omega} \in \mathcal D[\honeparticlehamiltonianUnshifed^{\widetilde u, \omega}]$. Thus, we have
		\begin{align} \label{Lemma 3.9 upper bound widehat eN1}
		\begin{split}
		\honeparticlehamiltonianUnshiftedEigenvalues^{1, \widetilde u, \widetilde k, \omega} & \le \braket{\phi_N^{1,\omega}}{\widehat h_N^{\widetilde u, \widetilde k, \omega} \phi_N^{1,\omega}}\\
		&\le \freelaplacianEigenvalueOne+N \int\limits_{\Lambda^{\omega}_{N}} \int\limits_{\Lambda^{\omega}_{N}} v_N(x-y) |u_N^{\widetilde k, \omega}(x)|^2|\phi_N^{1,\omega}(y)|^2\ \ud x \ud y   \\
		&\leq \freelaplacianEigenvalueOne+  N \|v_N\|_1 \|\phi_N^{1,\omega}\|^2_{\infty} \\
		& \leq \freelaplacianEigenvalueOne+  C_1^2 N \|v_N\|_1 (\freelaplacianEigenvalueOne)^{d/2}
		\end{split}
		\end{align}
		with $C_1 = 2 (4\pi)^{-d/4} \e$, where we made use of $\|u_N^{\widetilde k, \omega}\|_2 = 1$, \cite[Lemma~1.1]{SZNITMAN2023104197}, and the fact that $\phi_N^{1,\omega}$ has support only on $\Lambda_N^{\widetilde k_N^{\omega}}$, see Lemma~\ref{Lemma ground state of Dirichlet Laplace has support only on one single component} and Definition~\ref{definition widetilde j N omega}. So if $\freelaplacianEigenvalueTwo < \honeparticlehamiltonianUnshiftedEigenvalues^{1, \widetilde u, \widetilde k, \omega}$, then
		\begin{equation}
		\freelaplacianEigenvalueTwo < \freelaplacianEigenvalueOne + C_1 N \|\psi_N\|(\freelaplacianEigenvalueOne)^{ d/2} \ ,
		\end{equation}
		which contradicts our assumption that $\omega \in \Omega_N^{(2),\eta}$.
		
		To conclude, one has
		\begin{equation}
		\braket{\psi}{\honeparticlehamiltonianUnshifed^{\widetilde u, \omega} \psi} \ge \honeparticlehamiltonianUnshiftedEigenvalues^{1, \widetilde u, \widetilde k, \omega} \ .
		\end{equation}
		Now, note that the difference between $\honeparticlehamiltonianEigenvalues^{1,\widetilde u, \widetilde k, \omega}$ and $\honeparticlehamiltonianUnshiftedEigenvalues^{1, \widetilde u, \widetilde k, \omega}$ on the one hand and $\braket{\psi}{\honeparticlehamiltonian^{\widetilde u,\omega} \psi}$ and $\braket{\psi}{\honeparticlehamiltonianUnshifed^{\widetilde u,\omega} \psi}$ on the other hand is the same constant. Hence, one also has
		\begin{equation}
		\braket{\psi}{\honeparticlehamiltonian^{\widetilde u, \omega} \psi} \ge \honeparticlehamiltonianEigenvalues^{1,\widetilde u, \widetilde k, \omega} \ .
		\end{equation}

		Finally, since $u_N^{\widetilde k,\omega} \in H_0^1(\Lambda_N^{\omega})$, $\|u_N^{\widetilde k,\omega}\|_2 = 1$, and
		\begin{equation}
		\braket{u_N^{\widetilde k,\omega}}{\honeparticlehamiltonian^{\widetilde u, \omega} u_N^{\widetilde k,\omega}} = \honeparticlehamiltonianEigenvalues^{1,\widetilde u, \widetilde k, \omega} \ ,
		\end{equation}
		the statement follows.
	\end{proof}
	
	The next proposition, together with Proposition~\ref{proposition probability omegaN2 converges to one}, gives us a lower bound for the gap between the two lowest eigenvalues of the operator $\honeparticlehamiltonian^{\widetilde u, \omega}$ \eqref{definition klein widetilde h N u omega}. Recall that the eigenvalues are counted with multiplicity. This proposition also ensures that the ground state of $\honeparticlehamiltonian^{\widetilde u, \omega}$ is unique. The details are given in the subsequent Corollary~\ref{corollary unique ground state of h defined on all components}. Proposition~\ref{proposition estimate gap} and Corollary~\ref{corollary unique ground state of h defined on all components} are also crucial for the proofs of Theorems~\ref{theorem proof of BEC I} and~\ref{theorem proof of BEC II}.
	
	\begin{prop}\label{proposition estimate gap} Let an arbitrary $0 < \eta < \e^{-\nu \omega_d r^d}$ be given. Then for all $\omega \in \Omega_N^{(2)}$, we have
	
	\begin{align}
	\honeparticlehamiltonianEigenvalues^{2, \widetilde u, \omega} - \honeparticlehamiltonianEigenvalues^{1, \widetilde u, \omega} \geq \freelaplacianEigenvalueTwo - \freelaplacianEigenvalueOne - C_1^2 N \|v_N\|_1 (\freelaplacianEigenvalueOne)^{d/2} \ ,
		\end{align}
		where $C_1 = 2 (4\pi)^{-d/4} \e$, $\honeparticlehamiltonianEigenvalues^{1, \widetilde u, \omega}$ and $\honeparticlehamiltonianEigenvalues^{2, \widetilde u, \omega}$ are the two lowest eigenvalues of $\honeparticlehamiltonian^{\widetilde u, \omega}$ in $L^2(\Lambda_N^{\omega})$, see~\eqref{definition klein widetilde h N u omega}, and $\freelaplacianEigenvalueOne$ and $\freelaplacianEigenvalueTwo$ are the two lowest eigenvalues of the Dirichlet Laplacian $-\Delta$ on $\Lambda_N^{\omega}$.
    \end{prop}
	\begin{proof} Let an arbitrary $0 < \eta < \e^{-\nu \omega_d r^d}$, $N \in \mathds N$, and $\omega \in \Omega_N^{(2),\eta}$ be given.
	We begin by proving this statement for the operator $\honeparticlehamiltonianUnshifed^{\widetilde u, \omega}$ \eqref{definition klein widetilde h N u omega unshifted} first.
	
	We have $\honeparticlehamiltonianUnshiftedEigenvalues^{2, \widetilde u, \omega} \geq \freelaplacianEigenvalueTwo$, since $v_N \geq 0$ and $\mathcal D[\honeparticlehamiltonianUnshifed^{\widetilde u, \omega}] = \mathcal D[-\Delta]$. On the other hand, we have 
		\begin{align}
		\honeparticlehamiltonianUnshiftedEigenvalues^{1,\widetilde u, \omega}	\leq \freelaplacianEigenvalueOne+  C_1^2 N \|v_N\|_1\cdot (\freelaplacianEigenvalueOne)^{d/2}
		\end{align}
		with $C_1 = 2 (4\pi)^{-d/4} \e$, see \eqref{Lemma 3.9 upper bound widehat eN1}. Therefore,
		\begin{equation}
		 \honeparticlehamiltonianUnshiftedEigenvalues^{2, \widetilde u, \omega} - \honeparticlehamiltonianUnshiftedEigenvalues^{1, \widetilde u, \omega} \geq \freelaplacianEigenvalueTwo - \freelaplacianEigenvalueOne - C_1^2 N \|v_N\|_1 (\freelaplacianEigenvalueOne)^{d/2} \ .
		\end{equation}
		The argument that the difference between $\honeparticlehamiltonianEigenvalues^{2, \widetilde u, \omega}$ and $\honeparticlehamiltonianUnshiftedEigenvalues^{2, \widetilde u, \omega}$ on the one hand and $\honeparticlehamiltonianEigenvalues^{1, \widetilde u, \omega}$ and $\honeparticlehamiltonianUnshiftedEigenvalues^{1, \widetilde u, \omega}$ on the other hand is the same constant, and that this constant disappears for the expression $\honeparticlehamiltonianEigenvalues^{2, \widetilde u, \omega} - \honeparticlehamiltonianEigenvalues^{1, \widetilde u, \omega}$ now completes this proof.
	\end{proof}
	
	\begin{cor}\label{corollary unique ground state of h defined on all components}
		Suppose $0 < \eta < \e^{-\nu \omega_d r^d}$, $N \in \mathds N$, and $\omega \in \Omega_N^{(2),\eta}$. Then 
		\begin{equation} \label{Hauptungleichung Gap}
		\honeparticlehamiltonianEigenvalues^{2, \widetilde u, \omega} - \honeparticlehamiltonianEigenvalues^{1, \widetilde u, \omega} > 0
		\end{equation}
	\end{cor}
	\begin{proof}
		This follows immediately from Definition~\ref{definition Omega 2 event} and Proposition~\ref{proposition estimate gap}.
	\end{proof}

	\section{Proof of Bose--Einstein condensation} \label{section proof bec}
	In this section we shall use the results obtained in the previous sections in order to prove the occurrence of BEC in the interacting Bose gas in the Kac--Luttinger model for suitably scaled two-particle interactions.	
	In particular, we will show that the ground state of \eqref{N particle Hamiltonian}, denoted as $\Psi_N^\omega$, exhibits Bose--Einstein condensation in $u_N^{\widetilde k,\omega}$.
		Recall that $u_N^{\widetilde k, \omega}$ is a minimizer of the Hartree-type functional \eqref{HartreeFunctional componentwise} on the component $\Lambda_N^{\widetilde k_N^{\omega}}$ where the normalized eigenfunction corresponding to the lowest eigenvalue of $-\Delta$ has its support, see Definition~\ref{definition widetilde j N omega}.
	
	We define the one-particle density matrix associated with $\Psi_N^\omega$  as the nonnegative trace class operator on $L^2(\Lambda_N^\omega)$ with integral kernel
	\begin{equation}
	\varrho^{(1),\omega} (x;y) = \int dx_2 \dots dx_N \, \Psi^\omega_{N} (x , x_2, \dots , x_N) \overline{\Psi}^\omega_{N} (y,x_2, \dots , x_N)  \ ,
	\end{equation}
	which is normalized so that $\text{tr}\varrho^{(1),\omega}=1$ (see \cite{PO56,M07}). 
	We recall that the expectation of the number of particles occupying the state $u_N^{\widetilde k,\omega}$ is given by 
	\begin{equation}
	n_N^{\omega} \coloneqq N \cdot \mathrm{tr}(\varrho^{(1),\omega}\lvert u_N^{\widetilde k,\omega} \rangle \langle u_N^{\widetilde k,\omega} \lvert)
	\end{equation}
	for all $N \in \mathds N$, any $0 < \eta < \e^{-\nu \omega_d r^d}$, and all $\omega \in \Omega_N^{(1),\eta}$.  We set $n_N^{\widetilde k_N^{\omega}, \omega} \coloneqq 0$ for all $N \in \mathds N$, any $0 < \eta < \e^{-\nu \omega_d r^d}$, and all $\omega \notin \Omega_N^{(1),\eta}$.
	
	\begin{theorem} \label{theorem proof of BEC I}
		Let an arbitrary $0 < \eta < \e^{-\nu \omega_d r^d}$, $N \in \mathds N$, and $\omega \in \Omega_N^{(2),\eta}$ be given. Let $v_N$ be given with its Fourier transform $\widehat v_N \ge 0$. Let $E_{\text{QM},N}^{1,\omega}$ be the ground state energy of $H_{N}^{\omega}$, see \eqref{N particle Hamiltonian}, and $\honeparticlehamiltonianEigenvalues^{1, \widetilde u, \omega}$ the minimum of the Hartree-type functional, see \eqref{HartreeFunctional componentwise}. We then have
        \begin{equation}\label{EstimateGSenergy}
		\left|\frac{E_{\text{QM},N}^{1,\omega}}{N} - \honeparticlehamiltonianEigenvalues^{1, \widetilde u, \omega} \right| \leq \dfrac{v_N(0)}{2}
		\end{equation}
		and
		\begin{equation}\label{EstimateBesetzungszahl}
		1-\frac{n_N^{\omega}}{N} \leq \frac{v_N(0)}{2}\cdot \frac{1}{\honeparticlehamiltonianEigenvalues^{2, \widetilde u, \omega} - \honeparticlehamiltonianEigenvalues^{1, \widetilde u, \omega}}  \ ,
		\end{equation}
		where $v_N(0)=(2\pi)^{-d/2}\|\widehat v_N\|_{1}$.
	\end{theorem}
	
	\begin{proof}
		Let an arbitrary $0 < \eta < \e^{-\nu \omega_d r^d}$, $N \in \mathds N$, and $\omega \in \Omega_N^{(2),\eta}$ be given. For any $\xi \in L^1(\mathds{R}^d)$ we have \cite[Lemma~3.3]{lewin2015}
		\begin{equation}
		\sum_{1\leq i < j \leq N}v_N(x_i-x_j) \geq \sum_{j=1}^{N} (\xi \ast v_N)(x_j) - \frac{1}{2}\int\limits_{\mathds R^d} \int\limits_{\mathds R^d}v_N(x-y)\xi(x)\xi(y)\ \ud x \ud y - N \dfrac{v_N(0)}{2} \ .
		\end{equation}
		Setting $\xi(x)\coloneqq \sqrt{N(N-1)} |u_N^{\widetilde k, \omega}(x)|^2$ (understanding $u_N^{\widetilde k, \omega}$ to be extended by zero to $\Lambda_N^{\omega}$), we thus obtain
		\begin{equation}
		\begin{split}
		H_N^{\omega} \geq& \sum_{j=1}^{N}\Bigg(-\Delta_{j} + (N-1) (|u_N^{\widetilde k, \omega}|^2 \ast v_N)(x_j) \Bigg)\\
       & \qquad - \frac{N-1}{2}\int\limits_{\Lambda_N^{\omega}} \int\limits_{\Lambda_N^{\omega}} v_N(x-y)|u_N^{\widetilde k, \omega}(x)|^2|u_N^{\widetilde k, \omega}(y)|^2\ \ud x \ud y \Bigg) - N\frac{v_N(0)}{2} 
		\end{split}
		\end{equation}
		as an operator inequality on $L_{\text{s}}^2((\Lambda^{\omega}_{N})^N)$ and, with definition \eqref{definition klein widetilde h N u k omega} of $\honeparticlehamiltonian^{\widetilde u, \omega}$,
		\begin{equation}
		H_N^{\omega} \geq \sum_{j=1}^{N}\Big(\honeparticlehamiltonian^{\widetilde u, \omega}\Big)_j - N\frac{v_N(0)}{2} \ .
		\end{equation}
		Therefore,
		\begin{equation}\label{eq:lb}
		\frac{E_{\text{QM},N}^{1,\omega}}{N}\geq \mathrm{tr}(\varrho^{(1),\omega} \honeparticlehamiltonian^{\widetilde u, \omega})- \frac{v_N(0)}{2} \ .
		\end{equation}
		Recall here that $E_{\text{QM},N}^{1,\omega}$ is the lowest eigenvalue of the $N$-particle Hamiltonian $H_N^{\omega}$ defined in \eqref{N particle Hamiltonian}. 		
		Moreover, we have the upper bound
		\begin{align}\label{eq:ub}
		\begin{split}
		E_{\text{QM},N}^{1,\omega} & \le \braket{u_N^{\widetilde k, \omega} \otimes \ldots \otimes u_N^{\widetilde k, \omega}}{H_N^{\omega} u_N^{\widetilde k, \omega} \otimes \ldots \otimes u_N^{\widetilde k, \omega}} \\
		& = N \varepsilon_N^{1, \widetilde k, \omega} = N \honeparticlehamiltonianEigenvalues^{1, \widetilde u, \widetilde k, \omega} = N \honeparticlehamiltonianEigenvalues^{1,\widetilde u, \omega} \ ,
		\end{split}
		\end{align}
		where we used Proposition~\ref{Proposition ground state energy hNk equals componentwise hartree functional} and Lemma~\ref{lemma groundstate of h equals lowest componentwise h ground state}. Combining \eqref{eq:lb}, \eqref{eq:ub} and using that
		\begin{equation}
		\mathrm{tr}(\varrho^{(1),\omega} \honeparticlehamiltonian^{\widetilde u, \omega})\geq \honeparticlehamiltonianEigenvalues^{1, \widetilde u, \omega} \ ,
		\end{equation}
		inequality \eqref{EstimateGSenergy} follows.
		
		To prove \eqref{EstimateBesetzungszahl}, we observe that with \eqref{eq:lb} and \eqref{eq:ub} we obtain
		
		\begin{equation}
		\honeparticlehamiltonianEigenvalues^{1, \widetilde u, \omega} \geq \frac{n_N^{\omega}}{N} \honeparticlehamiltonianEigenvalues^{1, \widetilde u, \omega} + \left(1-\frac{n_N^{\omega}}{N}\right) \honeparticlehamiltonianEigenvalues^{2, \widetilde u, \omega} - \dfrac{v_N(0)}{2}
		\end{equation}
		and, also using Corollary~\ref{corollary unique ground state of h defined on all components},
		\begin{equation}
		1-\frac{n_N^{\omega}}{N} \leq \dfrac{v_N(0)}{2}\cdot \frac{1}{\honeparticlehamiltonianEigenvalues^{2, \widetilde u, \omega} - \honeparticlehamiltonianEigenvalues^{1, \widetilde u, \omega}} \ .
		\end{equation}
	\end{proof}
	
	We now state and prove our main result, namely, the occurrence of BEC in probability or with probability almost one under certain conditions for the pair-interaction $v_N$.
	
	\begin{theorem}[BEC] \label{theorem proof of BEC II} Suppose $0 < \eta < \e^{-\nu \omega_d r^d}$, and let $v_N$ together with its Fourier transform $\widehat v_N \ge 0$ for all $N \in \mathds N$ be given.
	\begin{enumerate}[(i)]
	\item For any $\varepsilon > 0$ there exists a constant $\kappa > 0$ such that if $\|v_N\|_1 \le \kappa N^{-1} (\ln N)^{-2/d}$ for all but finitely many $N \in \mathds N$ and $v_N(0)\ll (\ln N)^{-(1+2/d)}$, we have for any $\zeta>0$
		\begin{equation}
		\liminf\limits_{N \to \infty} \mathds P \left(\left| \dfrac{ n_N^{\omega}}{N} - 1 \right| < \zeta \right) \ge 1 - \varepsilon \ .
		\end{equation}
		This means, there is complete BEC with probability almost one into a minimizer of the Hartree-type functional \eqref{HartreeFunctional componentwise}.
		\item If $\|v_N\|_1  \ll N^{-1} (\ln N)^{-2/d}$ and $v_N(0)\ll (\ln N)^{-(1+2/d)}$, where $v_N(0)=(2\pi)^{-d/2}\|\widehat v_N\|_{1}$ and $\widehat v_N$ is the Fourier transform of $v_N$, then for all $\zeta > 0$ we have
		\begin{equation}
 		\lim\limits_{N \to \infty} \mathds P \left(\left| \dfrac{ n_N^{\omega}}{N} - 1 \right| < \zeta \right) = 1 \ ,
        \end{equation}
        that is, there is complete BEC in probability into a minimizer of the Hartree-type functional \eqref{HartreeFunctional componentwise}.
        \end{enumerate}
	\end{theorem}
	\begin{proof} 
		Let an $0 < \eta < \e^{-\nu \omega_d r^d}$, $N \in \mathds N$, and $\omega \in \Omega_N^{(2),\eta}$ be given. Then by Theorem~\ref{theorem proof of BEC I} and Proposition~\ref{proposition estimate gap}, we have
		\begin{equation} \label{Equation Theorem 4.2 asdasdasd234glk}
		1 - \dfrac{n_N^{\omega}}{N} \le \dfrac{v_N(0)}{2}\cdot \frac{1}{\freelaplacianEigenvalueTwo - \freelaplacianEigenvalueOne - C_1^2 N \|v_N\|_1 (\freelaplacianEigenvalueOne)^{d/2}}
		\end{equation}
		
		Note that for the gap between the two lowest eigenvalues $\freelaplacianEigenvalueOne$ and $\freelaplacianEigenvalueTwo$ of the Dirichlet Laplacian on $\Lambda_N^{\omega}$ \cite[Theorem~6.1]{SZNITMAN2023104197} one has
		\begin{equation} \label{Equation Theorem 4.2 asdopizxcio}
		\lim\limits_{\sigma \to 0} \liminf\limits_{N \to \infty} \mathds P \left( \freelaplacianEigenvalueTwo - \freelaplacianEigenvalueOne \ge \sigma (\ln N)^{-(1+2/d)} \right) = 1 \ .
		\end{equation}
		In addition, recall that $\lim_{N \to \infty} \mathds P ( \Omega_N^{(3)} ) = 1$ where 
		\begin{equation}
		 \Omega_N^{(3),\eta} = \left\{ \omega \in \Omega_N^{(1),\eta} :  e_N^{1,\omega} \le C_2 (\ln N)^{-2/d} \right\} \ . 
		\end{equation}
		We define $C_3 \coloneqq C_1^2 C_2^{d/2}$.

		We firstly discuss the case $(i)$. Let an arbitrary $\varepsilon > 0$ be given. By \eqref{Equation Theorem 4.2 asdopizxcio}, there exists a $\sigma > 0$ such that for all but finitely many $N \in \mathds N$ we have $\mathds P(\Omega_N^{(4),\eta,\sigma}) \ge 1 - \varepsilon/2$ where
		\begin{equation}
		\Omega_N^{(4),\eta,\sigma} \coloneqq \left\{ \omega \in \Omega_N^{(1), \eta} :  \freelaplacianEigenvalueTwo - \freelaplacianEigenvalueOne \ge \sigma (\ln N)^{-(1+2/d)} \right\} \ , \quad N \in \mathds N \ .
		\end{equation}
		Therefore, if $\|v_N\|_1 < \sigma C_3^{-1} N^{-1} (\ln N)^{-2/d}$ for all but finitely many $N \in \mathds N$ and $v_N(0)\ll (\ln N)^{-(1+2/d)}$, we have
		\begin{align} 
		& \, \mathds P \left( \left| \dfrac{ n_N^{\omega}}{N} - 1 \right| < \zeta  \right)\\
		\ge \, & \mathds P \left(  \left| \dfrac{ n_N^{\omega}}{N} - 1 \right| \le \dfrac{v_N(0)}{2[\sigma (\ln N)^{-(1 + 2/d)} - C_3 N \|v_N\|_1 (\ln N)^{-1}] } \right) \\
		\ge \, & \mathds P \left( \left\{ \omega \in \Omega_N^{(2),\eta} : 1 - \dfrac{n_N^{\omega}}{N} \le \dfrac{v_N(0)}{2[\freelaplacianEigenvalueTwo - \freelaplacianEigenvalueOne - C_1^2 N \|v_N\|_1 (e_N^{1,\omega})^{d/2}] } \right\} \right. \\
		& \qquad \cap \left. \Omega_N^{(3),\eta} \cap \Omega_N^{(4),\eta,\sigma} \right) \\
		= \, & \mathds P \left( \Omega_N^{(2),\eta} \cap \Omega_N^{(3),\eta} \cap \Omega_N^{(4),\eta,\sigma} \right) \\
		\ge \, & \mathds P ( \Omega_N^{(2),\eta}) + \mathds P( \Omega_N^{(3),\eta}) + \mathds P(\Omega_N^{(4),\eta,\sigma} ) - 2
		\end{align}
		for any $\zeta > 0$ and for all but finitely many $N \in \mathds N$. Note that we used \eqref{Equation Theorem 4.2 asdasdasd234glk} for the last step. By Proposition~\ref{proposition probability omegaN2 converges to one}, there exists a $\kappa$ such that if $\|v_N\|_1 \le \kappa N^{-1} (\ln N)^{-2/d}$ for all but finitely many $N \in \mathds N$, then $\mathds P(\Omega_N^{(2),\eta}) \ge 1 - \epsilon/2$ for all but finitely many $N \in \mathds N$. Therefore, for any $\zeta > 0$ we have
		\begin{equation}
		\liminf\limits_{N \to \infty} \mathds P \left( \left| \dfrac{ n_N^{\omega}}{N} - 1 \right| < \zeta  \right) \ge 1 - \varepsilon \ .
		\end{equation}
		
		Lastly, for the case $(ii)$ we conclude from \eqref{Equation Theorem 4.2 asdopizxcio} that
		$\lim_{N \to \infty} \mathds P(\Omega_N^{(4),\eta, \sigma_N}) = 1$ where
		\begin{equation}
		\Omega_N^{(4),\eta, \sigma_{N \in \mathds N}} \coloneqq \left\{ \omega \in \Omega_N^{(1), \eta} :  \freelaplacianEigenvalueTwo - \freelaplacianEigenvalueOne \ge \sigma_N (\ln N)^{-(1+2/d)} \right\} \ , \quad N \in \mathds N
		\end{equation}
		and $(\sigma_N)_{N \in \mathds N}$ is an arbitrary sequence that converges to zero. Therefore, for any $\varepsilon > 0$ and any sequence $\sigma_{N \in \mathds N}$ that converges to zero such that $\sigma_N \gg v_N(0) (\ln N)^{1+2/d}$ and for which we have, for some $0 < \tilde \varepsilon <1$,
		\begin{equation}
		\sigma_N > (1-\tilde \varepsilon)^{-1} C_3 (\ln N)^{2/d} N \|v_N\|_1
		\end{equation}
		for all but finitely many $N \in \mathds N$, we conclude, similarly as above,
		\begin{align} 
		& \, \mathds P \left( \left| \dfrac{ n_N^{\omega}}{N} - 1 \right| < \zeta  \right)\\
		\ge \, & \mathds P \left(  \left| \dfrac{ n_N^{\omega}}{N} - 1 \right| \le \dfrac{v_N(0)}{2[\sigma_N (\ln N)^{-(1 + 2/d)} - C_3 N \|v_N\|_1 (\ln N)^{-1}] } \right) \\
		= \, & \mathds P \left( \Omega_N^{(2),\eta} \cap \Omega_N^{(3),\eta} \cap \Omega_N^{(4),\eta,(s_N)_{N \in \mathds N}} \right)
		\end{align}
		for any $\zeta > 0$ and all but finitely many $N \in \mathds N$.
		%We again used \eqref{Equation Theorem 4.2 asdasdasd234glk} for the last step.
		Since the right side of this inequality now converges to one in the limit $N \to \infty$, see also Proposition~\ref{proposition probability omegaN2 converges to one}, we have shown that for all $\zeta > 0$,
		\begin{equation}
		\lim\limits_{N \to \infty} \mathds P \left( \left| \dfrac{ n_N^{\omega}}{N} - 1 \right| < \zeta \right) = 1 \ .
		\end{equation}
	\end{proof}
	
	\begin{remark} It is possible to relax the assumption of $v_N(0) \ll (\ln N)^{-(1+2/d)}$ to $v_N(0) \le c_1 (\ln N)^{-(1+2/d)}$ for all but finitely many $N \in \mathds N$ for a sufficiently small constant $c_1 > 0$ and still conclude, similarly as in the proof of Theorem~\ref{theorem proof of BEC II} that
	for a certain constant $c_2 > 0$ and for any $\zeta>0$ and $\varepsilon > 0$,
	\begin{equation}
		\liminf\limits_{N \to \infty} \mathds P \left(\left| \dfrac{ n_N^{\omega}}{N} - c_2 \right| < \zeta \right) \ge 1 - \varepsilon
    \end{equation}
    and
    \begin{equation}
		\lim\limits_{N \to \infty} \mathds P \left(\left| \dfrac{ n_N^{\omega}}{N} - c_2 \right| < \zeta \right) \ge 1 \ ,
    \end{equation}
	respectively. That is, one can still show the occurrence of BEC with probability almost one or in probability, although the condensation may not be complete anymore.
	\end{remark}

	\begin{remark}
		It is interesting to compare Theorem~\ref{theorem proof of BEC II} with [Theorem~4.2,\cite{kerner_pechmann_2023}] which makes a statement about the absence of BEC for suitably scaled repulsive two-particle interaction, at positive temperatures $T > 0$ (for completeness, one should mention that the authors focus in \cite{kerner_pechmann_2023} on the nonpercolation regime, which means that the intensity of the Poisson point process is chosen large enough). Assuming a potential $v_N(x):=w_N(\|x\|)$ with $w_N:\mathds{R} \rightarrow [0,\infty)$ such that
		\begin{equation}
		w_N(\|x\|) \ge b_N \quad \text{for} \quad \|x\| \le a_N \ ,
		\end{equation}
		it has been proved in \cite{kerner_pechmann_2023} that no one-particle state supported only on a single component (such as the minimizer of the Hartree-type functional considered in Section~\ref{section hartree functional}) is almost surely not macroscopically occupied if
		\begin{equation}
		\lim_{N \rightarrow \infty} \frac{b_N (a_N)^{3d} N}{(\ln N)^3}=\infty \quad \text{ and } \quad \lim_{N \rightarrow \infty} \frac{(a_N)^{3d} N}{( \ln N)^3}=\infty \ ,
		\end{equation}
		where $(a_N)_{N \in \mathds N}$ is a bounded sequence, as well as $\|v_N\|_1 \ll (\ln N)^{-2}$ (the last assumption was needed in \cite{kerner_pechmann_2023} to ensure that the physical system is well defined). Fixing $a_N = const.$ for all $N \in \mathds N$, for example, an expected regime for which BEC into a localized state is therefore possible only if $b_N \lesssim (\ln N)^3/N$. So whenever  $\|v_N\|_1 \sim b_N$, the condition on $v_N$ as formulated in Theorem~\ref{theorem proof of BEC II} seem quite close to being optimal. On the other hand, there is still some intermediate scaling regime that needs to be addressed in the future.
		
		Also, for us it seems possible that for stronger two-particle interactions, there might still be a macroscopic occupation of a one-particle state but of one that is not too localized.
	\end{remark}
	
	\begin{remark}
	It is also interesting to compare the result of Theorem \ref{theorem proof of BEC II} with related ones regarding BEC in nonrandom models. In the case of interacting bosons trapped in a region of order one, referring to a region independent of $N$, BEC has been proved to occur in mean-field models where the interaction scales as $v_N \sim N^{-1}V(x)$ where $V$ is a nonnegative function $V: \mathds R^d \to \mathds R$ independent of $N$ \cite{GrechSeiringer2013, lewin2015}. In addition, BEC is also present in the so-called Gross-Pitaevskii regime where, in three dimensions, the potential scales as $v_N \sim N^{2}V(Nx)$; here $V: \mathds R^3 \to \mathds R$ is again a nonnegative function independent of $N$  \cite{LiebSeiringer2002, lieb2005mathematics}. In both settings, the condensate wave function is a minimizer of a suitable one-particle functional similar to \eqref{HartreeFunctional componentwise}. 
	 
	 However, since the particles in our system are not confined in a region of order one, it may be better to compare the interaction strength of our system to the interaction strengths of these nonrandom models after rescaling the lengths accordingly. It may be reasonable to scale to a system where the particles are confined in a region $\Lambda_N= (-L_N/2,+L_N/2)^d$ where $L_N = \rho^{-1} N^{1/d}$. On the other hand, since the volume of the largest component in our system is at most of order $\ln N$ (at least in the nonpercolation regime), it may be more appropriate to compare our system to a nonrandom one where the particles are confined in a region $\Lambda_N= (-L_N/2,+L_N/2)^d$ where $L_N = (const.) (\ln N)^{1/d}$.

	 Lastly, an effect of randomness in our interacting model is the localization of the condensate wave function in a relatively small region, which is determined by the lowest eigenfunction of the Laplacian. In this context, we also would like to refer the reader to another nonrandom model studied in \cite{RouSpeh2023} where bosons separated by a double well potential display a localized regime.
	\end{remark}
	 
	\subsubsection*{Acknowledgements} We would like to thank the Mathematical Institute Oberwolfach (MFO) for hospitality during the Mini-Workshop \textit{A Geometric Fairytale full of Spectral Gaps and Random Fruit}. We also thank the other participants for interesting and stimulating discussions. In particular, we would like to thank Alain-Sol Sznitman as well as Christian Brennecke and Serena Cenatiempo for stimulating discussions on Bose--Einstein condensation. C. B. gratefully acknowledges funding from the Italian Ministry of University and Research (MIUR) through the PRIN 2022 project PRIN202223CBOCC\_01, project code 2022AKRC5P.
	
	{\small
		\bibliographystyle{amsalpha}
		\bibliography{Literature}}
	
\end{document}